\documentclass[submission,copyright,creativecommons]{eptcs}
\providecommand{\event}{GandALF 2025} %
\usepackage{tikz}
\usetikzlibrary{intersections, calc, arrows, automata, positioning,shapes, trees}
\usepackage{bbding}
\usepackage{pifont}
\usepackage{pgf, verbatim}
\usepackage{amssymb}
\usepackage{amsmath,amsthm}
\usepackage{tikz}
\usetikzlibrary{bending,arrows.meta}
\usetikzlibrary{decorations.pathmorphing}
\usepackage{comment}
\usepackage{complexity}

\newtheorem{thm}{Theorem}
\newtheorem{expl}{Example}
\newtheorem{col}{Corollary}
\newtheorem{remark}{Remark}

\newcommand{\game}{{\mathcal{G}}}
\newcommand{\player}[1]{\mathsf{Player}(#1)}

\newcommand{\Ad}{{\mathsf{Adam}}}
\newcommand{\Aa}{{\mathcal{A}}}
\newcommand{\Ev}{{\mathsf{Eve}}}
\newcommand{\Gg}{{\mathcal G}} 
 
\newcommand{\nat}{{\mathbb N}} %
\newcommand{\F}{{\mathcal{F}}}

\newcommand{\sg}{\sigma}

\newcommand{\AtE}{\mathsf{Attr}_{\Ev}}

\newcommand{\genreach}{\mathsf{GenReach}}
\newcommand{\maxgenreach}{\mathsf{MaxGenReach}}

\newcommand{\maxgenreachpromise}{\mathsf{MaxGenReachPromise}}

\usepackage{tcolorbox}
\usepackage[capitalise]{cleveref}
\usepackage{lineno}
\usepackage{todonotes}
\usepackage{iftex}

\ifpdf
  \usepackage{underscore}         %
  \usepackage[T1]{fontenc}        %
\else
  \usepackage{breakurl}           %
\fi

\title{Generalised Reachability Games Revisited%
\thanks{This work was supported by the EPSRC through grants EP/Z003121/1, EP/X03688X/1, EP/X042596/1 and EP/V025848/1 and by the Fonds de la Recherche Scientifique – FNRS under Grant n° T.0188.23 (PDR ControlleRS).
}}
\author{Sougata Bose
\institute{UMONS - Universit\'e de Mons\\ Mons, Belgium}
\email{\quad sougata.bose@umons.ac.be}
\and
Daniel Hausmann
\institute{University of Liverpool\\ Liverpool, UK}
\email{\quad d.hausmann@liverpool.ac.uk}
\and 
Soumyajit Paul
\institute{University of Liverpool\\ Liverpool, UK}
\email{\quad soumyajit.paul@liverpool.ac.uk}
\and 
Sven Schewe
\institute{University of Liverpool\\ Liverpool, UK}
\email{\quad sven.schewe@liverpool.ac.uk}
\and
Tansholpan Zhanabekova
\institute{University of Liverpool\\ Liverpool, UK}
\email{\quad t.zhanabekova@liverpool.ac.uk }%
}

\def\titlerunning{Generalised Reachability Games Revisited}
\def\authorrunning{S.\ Bose, D.\ Hausmann, S.\ Paul, S.\ Schewe, T.\ Zhanabekova }
\begin{document}

\maketitle

\begin{abstract}
Classic reachability games on graphs are zero-sum games, where the goal of one player, $\Ev$, is to visit a vertex from a given target set, and that of other player, $\Ad$, is to prevent this. Generalised reachability games, studied by Fijalkow and Horn, are a generalisation of reachability objectives, where instead of a single target set, there is a family of target sets and $\Ev$ must visit all of them in any order. In this work, we further study the complexity of solving two-player games on graphs with generalised reachability objectives. Our results are twofold: first, we provide an improved complexity picture for generalised reachability games, expanding the known tractable class from games in which all target sets are singleton to additionally allowing a logarithmic number of target sets of arbitrary size. Second, we study optimisation variants of generalised reachability with a focus on the size of the target sets. For these problems, we show intractability for most interesting cases. Particularly, in contrast to the tractability in the classic variant for singleton target sets, the optimisation problem is \textsc{NP}-hard when $\Ev$ tries to maximise the number of singleton target sets that are visited. Tractability can be recovered in the optimisation setting when all target sets are singleton by requiring that $\Ev$ pledges a maximum sized subset of target sets that she can guarantee to visit.

\end{abstract}

\section{Introduction}

Two-player zero-sum games played on graphs provide a fundamental framework for modelling decision-making scenarios where two players have opposing objectives. They are extensively used to model reactive systems, where one player ($\Ev$) models actions controlled by the system and the other player ($\Ad$) models uncontrollable actions of the environment. Analysing such games then provides formal guarantees on the behaviours of the system against all possible behaviours of the environment~\cite{fijalkow2023gamesgraphs,pnueli1977, PR89, Bloem2018}.

These games are played on finite graphs, where the vertices are partitioned into ones controlled by $\Ev$ and $\Ad$. Starting from a token placed on a fixed initial vertex, the player controlling the current vertex moves the token along an edge of the graph to jointly form an infinite path. A winning objective specifies the set of acceptable behaviours of the system as a set of infinite paths that are good for $\Ev$. In the game, $\Ev$ attempts to ensure that the path formed is good for $\Ev$, while $\Ad$ tries to obstruct this goal.
Solving games refers to the decision problem of checking which of the two player can win from a given initial vertex.

A key type of objective studied in such settings is \emph{reachability}, where the goal of $\Ev$ is to reach some vertex among a designated target set of vertices. 
Generalised reachability games extend this concept by requiring $\Ev$ to reach multiple target sets rather than just one. Then, $\Ev$ is required to visit at least one vertex from each target set in the play.
However, in terms of computational complexity of solving such games, changing from reachability to generalised reachability results in a significant jump. While reachability games are P-complete \cite{Zer13,DBLP:journals/jcss/Immerman81}, generalised reachability games are \PSPACE-complete~\cite{fij-horn2010}. 
Understanding the complexity of solving special cases of generalised reachability games is crucial, especially when considering the size of target sets as a parameter. For instance, if all target sets consist of a single vertex, $\Ev$ is required to visit several vertices in the play. Such games can be solved in polynomial time~\cite{fij-horn2010}.

The analysis of generalised winning conditions, i.e., the conjunction of multiple objectives of the same kind has been a significant focus of research~\cite{chatterjee2007, chatterjee_et_al10, chatterjee_et_al16}.  
While generalised reachability games have been well studied, one could also consider the optimisation variant of the problem. This asks $\Ev$
to visit as many target sets as possible.
This variant can be used to provide strategies where visiting all target sets may not be possible, but $\Ev$ can still visit a significant number of target sets.
Optimisation variants have been studied beyond generalised reachability objectives~\cite{KupfermanS24}.
In this work, we consider both the case where $\Ev$ has to name the target sets she visits before the game starts (and visiting other sets does not count) \emph{and} the case where she just wants to maximise the number of sets visited.

\paragraph{Our Contributions.}
In this paper, we provide several new complexity results for generalised reachability games:

\begin{itemize}
    \item We prove that generalised reachability can be checked in time linear in (1) the size of the game, (2) the number of singleton target sets, and (3) exponential in the number of larger target sets.
    It is therefore fixed-parameter tractable (FPT) when considering the number of target sets larger than one as the parameter and polynomial if the number of large target sets is logarithmically bounded.
    \item We establish \NL-completeness for the single-player case, where only the environment makes decisions.
    \item We analyse optimisation variants of the problem, showing that \emph{maximising} the number of singleton target sets (or: target states) visited depends on whether we want to first name the visited target states or to just maximise the number of visited states; they are tractable and NP-complete, respectively.
\end{itemize}

By investigating these computational aspects, we provide a deeper understanding of the complexity landscape of generalised reachability games and contribute to the broader field of game theory and formal verification.

\paragraph*{Related Works.}

The starting point of this work is the study of \emph{generalised reachability games} by Fijalkow and Horn~\cite{fij-horn2010}.
Their work revealed the surprising complexity of these games: while reachability games are \textsc{P}-complete \cite{Zer13,DBLP:journals/jcss/Immerman81}, even modest extensions to standard reachability objectives can lead to significant computational challenges. 
This study has influenced the analysis of multi-objective games and strategy synthesis under complex constraints. 
This work also provides certain restrictions which make the problem easier, particularly by considering one-player variants, and by parametrising the problem by size of each target set.
They also exploit connections with the true quantified satisfiability (QSAT) problem, one of the standard \PSPACE-complete problems, to provide lower bounds, both in the general case and in some restricted cases. 
An open problem stated in their work is the complexity of the problem when target sets have size $2$. 

Games with weighted multiple objectives have been studied by Kupferman and Shenwald~\cite{KupfermanS24}, showing how adding different goals and weights makes these games more complex and useful for modelling real systems.
In particular, their results can be used to provide strategies that maximise the number of objectives satisfied, even if not all of the objectives can be jointly satisfied. This yields a \PSPACE~upper bound for some of the optimisation variants of generalised reachability games we consider. However, they do not study the problem by considering the size of target sets as a parameter.

In the realm of satisfiability, El Halaby~\cite{Halaby2016} investigated the \emph{computational complexity of MaxSAT}, an optimisation variant of the satisfiability problem. 
Prior to this, Kohli et al.~\cite{Kohli-et-al-94} introduced and studied the \emph{Minimum Satisfiability Problem (MinSAT)}, establishing its NP-hardness and discussing its relevance in fields such as fault diagnosis and design verification. However, the optimisation variant of QSAT, a natural \PSPACE-complete problem, is relatively unexplored. To the best of our knowledge, while some algorithms for solving MAX-QSAT are known~\cite{IgnatievJM13}, the computational complexity for restricted classes of the problem have not been studied. As several lower bound proofs for generalised reachability games are obtained by reduction from QSAT, we expect that the analysis of optimisation variants of generalised reachability games provides more insight into the optimisation variants of QBF and vice-versa.  
\section{Preliminaries}\label{sec:pre}
We use $\nat$ to denote the set of natural numbers. For $n \in \nat$, we use $[n]$ to denote the set $\{1,\dots,n\}$. We use $G = (V,E)$ to denote a directed graph with sets of vertices $V$ and edges $E \subseteq V \times V$. We often write $u \rightarrow v$ to denote $(u,v) \in E$. For $u \in V$, let $Succ(u) = \{v\in V \mid u \rightarrow v\}$. We assume familiarity with graph theoretic notions such as strongly connected components (SCC) and the directed acyclic graph formed by SCC decomposition of a graph.%

In this work we consider two-player turn-based games played between players $\Ad$ and $\Ev$. Such games are played on directed graphs called game arenas. Formally, a \emph{game arena} $\Aa = (G, V_{\Ev}, V_{\Ad})$ is composed of a finite directed graph $G = (V, E)$ and a partition $(V_{\Ev}, V_{\Ad})$ of the vertex set $V$. A vertex in $V_{\Ev}$ (respectively $V_{\Ad}$) is controlled by $\Ev$ (respectively $\Ad$). For $v \in V$, let $\player{v}$ be the player controlling $v$ i.e. $\player{v} = p$ when $v \in V_p$ for $p \in \{\Ad, \Ev\}$.

A \emph{generalised reachability game} is a tuple $\game = (\Aa, s,\F)$, where
\begin{itemize}
\item $\Aa$ is a game arena
\item $s \in V$ is the start vertex
\item $\F = \{F_1, F_2, \dots, F_n\}$ is a set of $n$ target sets, where for each $i$, we have $F_i \subseteq V$.
\end{itemize}

\cref{fig:maxproblem} is an example of a generalised reachability game where circle nodes belong to $\Ev$ and square nodes belong to $\Ad$. $s$ is the start vertex. In this particular game, all target sets are singleton and are marked with doubled circles, i.e. $\F = \{\{u_1\},\{u_2\},\{u_3\},\{u_4\},\{v\}\}$.

The rules of a generalised reachability game are as follows: initially, a token is placed on the start vertex $s$. At each step, when the token is on vertex $u$, it is $\player{u}$'s turn to play: $\player{u}$ chooses a vertex $v$ from $Succ(u)$ and moves the token to $v$. The sequence of vertices starting with $s$, that is obtained this way, is called a \emph{play}. 
The objective of $\Ev$ is to move the token in a way, such that the token visits some vertex from every target set $F_i \in \F$.
$\Ad$'s goal is to prevent $\Ev$ from achieving her objective.
In order to achieve their respective goals, players can move tokens according to some \emph{strategy}.
When the token is on vertex $u$, $\player{u}$ moves the token based on the past history as described by the strategy. 
Formally,  a \emph{strategy} for player $p \in \{\Ad,\Ev\}$ is a function $ \sg_p : V^* \cdot V_{p} \to V $, such that for all $\pi=v_0 v_1 \ldots v_k\in V^* \cdot V_{p}$, we have $\sg_p(\pi)\in Succ(v_k)$.
Thus a strategy prescribes a valid move that should be taken when it is the respective player's turn.
A pair of strategies \( (\sigma_{\text{Eve}}, \sigma_{\text{Adam}}) \) induces a unique infinite play $\pi \in V^{\omega}$.

In principle, the game can continue for an infinite duration but for generalised reachability objectives, player $\Ev$ has 
to achieve all her goals in a finite number of game steps.
Given a generalised reachability objective $\F = \{F_1, F_2, \dots, F_n\}$,
a play $\pi=v_0 v_1\ldots v_k$ is \emph{winning} for $\Ev$ if it visits each target set at least once, that is,
if for all $1\leq i\leq n$, there is $0\leq j\leq k$ such that $v_j\in F_i$; otherwise, $\pi$ is won by $\Ad$.

Our primary interest in this paper is the complexity of deciding whether $\Ev$ has a strategy to achieve the generalised reachability objective, starting
from $s$. The decision problem is as follows:

\begin{tcolorbox}[colback=gray!5!white,colframe=black!75!black,arc=0pt,outer arc=0pt]
$\boldsymbol{\genreach}$: Given a game $\game = (\Aa, s, \F)$, does $\Ev$ have a strategy to visit all $F_i$ in $\F$ starting from $s$?
\end{tcolorbox}

\begin{thm}[\cite{fij-horn2010}]
$\genreach$ is \PSPACE-complete. The \PSPACE-hardness holds even when $|F_i| = 3$ for each $F_i \in \F$. $\genreach$ is in $\P$ when $|F_i| = 1$ for each $F_i \in \F$.
\end{thm}

In this work, we particularly focus on the complexity of $\genreach$ parametrised by the size of target sets. This includes cases where some target sets may have size $1$. In this case, we simplify the notation for convenience. The singleton target sets are given by a set $T = \{t_1,\dots,t_m\} \subseteq V$. For such games, we write $\game = (\Aa,s,\F,T)$ where $|F_i| > 1$ for each $F_i \in \F$. In terms of the previous formulation, the set of targets is now $\F \cup \{\{t\}\mid t \in T\}$.

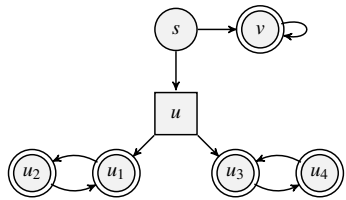
\begin{figure}[h!]
 \centering
\scalebox{0.7}{
\begin{tikzpicture}[shorten >=0.7pt, node distance=1.6cm, on grid, auto, thick, >=stealth']
   \tikzstyle{state}=[circle, draw, fill=gray!10, minimum size=8mm]
   \tikzstyle{box}=[rectangle, draw, fill=gray!10, minimum size=8mm]

   \node[state] (s) [yshift=-0.8cm] {$s$};
   \node[box] (u) [below =of s] {$u$};
   \node[state,double, double distance=0.7mm] (u1) [below left=of u] {$u_1$};
   \node[state,double, double distance=0.7mm] (u2) [left=of u1] {$u_2$};
   \node[state,double, double distance=0.7mm] (u3) [below right=of u] {$u_3$};
   \node[state,double, double distance=0.7mm] (u4) [right=of u3] {$u_4$};
   \node[state,double, double distance=0.7mm] (v) [right=of s] {$v$};

   \path[->]

   (s) edge (u)
   (u) edge (u1)
   (u) edge (u3)
   (s) edge (v)
   (u1) edge [bend right=30] (u2)
   (u2) edge [bend right=30] (u1)
   (u3) edge [bend right=30] (u4)
   (u4) edge [bend right=30] (u3)
   (v) edge [loop right] (v);;
\end{tikzpicture}
}
\caption{Generalised reachability game with all singleton targets, $T = \{u_1,u_2,u_3,u_4,v\}$}
\label{fig:maxproblem}
\end{figure}
In this paper, we also consider optimisation versions of the problem. In the game in \cref{fig:maxproblem}, all  targets are singleton with $T = \{u_1,u_2,u_3,u_4,v\}$ and $\Ev$ does not have a winning strategy from $s$.  However, at $s$ if she chooses to move to $u$, no matter what $\Ad$ chooses at $u$, at least two target sets will be visited. This is the best $\Ev$ can do, since on choosing $v$, she can only visit one target set. We consider the relevant decision problem, where $\Ev$'s goal is to maximise the number of target sets she can visit defined as follows: 

\begin{tcolorbox}[colback=gray!5!white,colframe=black!75!black,arc=0pt,outer arc=0pt]
$\boldsymbol{\maxgenreach}$: Given a game $\game$ and $ k \in \nat$, does $\Ev$ have a strategy to visit at least $k$ of the sets from $\F$ starting from $s$?
\end{tcolorbox}

In \cref{fig:maxproblem}, $\Ev$ can force the play to visit at least two target sets, but she cannot force to visit two specific target sets, since $\Ad$ makes the decision at $u$. On the other hand, $\Ev$ has a strategy to ensure that $v$ is always visited. %
We consider the optimisation variant, where the goal of $\Ev$ is to find the biggest collection of target sets, such that she can visit each of them in the collection.
We consider the relevant decision problem defined as follows: 
\begin{tcolorbox}[colback=gray!5!white,colframe=black!75!black,arc=0pt,outer arc=0pt]
$\boldsymbol{\maxgenreachpromise}$: Given a game $\game$, and $k\in \nat$, is there a set $\F' \subseteq \F$, with $|\F'| \geq k$  such that $\Ev$ has a strategy to visit all $F_i$ in $\F'$ starting from $s$?
\end{tcolorbox}

For the one-player variants of the game,  where only $\Ev$ plays%
, the  problems $\maxgenreach$
and $\maxgenreachpromise$ are equivalent. But this need not be true in general, even for the one-player variant with $\Ad$ as the only player. To see this, consider a game obtained by modifying the game in \cref{fig:maxproblem}, where all vertices in the game belong to $\Ad$. The maximum number of target sets $\Ev$ can ensure for $\maxgenreach$ is $1$, whereas for the $\maxgenreachpromise$, the maximum is $0$, since there is not even one target in $T$ which $\Ev$ can force to visit. 
$\Ad$ visits $v$ if $v$ is not promised by $\Ev$, otherwise goes to $u$ followed by a loop, say $u_1\to u_2$.

Reachability objectives require $\Ev$ to be able to enforce the play to enter target vertices, irrespective of what $\Ad$ plays. In this regard, the notion of \emph{attractor} is natural~\cite{Zer13,ZIELONKA1998135}.
For a set of (target) vertices $S$, the \emph{attractor} of $S$ is the set of all starting vertices $u$ such that $\Ev$ has strategy to reach some vertex in $S$ when the game starts at $u$. Formally, the attractor of $S$, denoted by $\AtE(S)$ can be defined recursively as follows. Here $\AtE^i(S)$ is the set of all starting vertices, from where $\Ev$ can force the play to enter $S$ within $i$ steps.
\begin{align*}
\AtE^0(S) &= S \\
\AtE^{i+1}(S) &= \AtE^i(S) ~ \cup~ \{u \in V_{\Ev} \mid \exists v \in Succ(u),  v \in \AtE^i(S) \} \\ &~~~~ \cup~  \{u \in V_{\Ad} \mid \forall v \in Succ(u),v \in \AtE^i(S) \} \\
\AtE(S) &= \bigcup_i \AtE^i(S)
\end{align*}

We point out that attractor sets $\AtE(S)$ can be computed in time linear in the number of edges of the underlying graph $(V,E)$, that is, in time $\mathcal{O}(|E|)$, noting $|E|\leq |V|^2$.
For each vertex $v\in\AtE(S)$, player $\Ev$ can enforce that $S$ is visited when playing from $v$: for $v\in V$, let $i_v$ denote the least 
number such that $v\in \AtE^{i_v}(S)$. Then $\Ev$ intuitively can ensure that $S$ is reached from $v$ in at most $i_v$ steps. A witnessing strategy is obtained by
moving from $v$ to some $v'\in\AtE^{i_v-1}(S)\cap Succ(v)$, that is, by moving one step closer to $S$.

\section{Solving Generalised Reachability Games}
We first establish \PSPACE-hardness of solving two-player generalised reachability games even when the underlying graph is a directed acyclic graph (DAG) with pathwidth 2.
To this end, we recall the hardness reduction of Fijalkow and Horn \cite{fij-horn2010}
that reduces the \PSPACE-complete TQBF problem to arbitrary generalised reachability games.
The input of the reduction is a quantified Boolean formula 
\[\phi=\forall x_1.\,\exists x_2. \forall x_3.\ldots. \forall x_n.\exists y_n.\, c_1\land\ldots\land c_n\]
where the $c_i$ are sets of (possibly negated) literals, indicating disjunctive clauses. The reduced game is the generalised reachability game $G_\phi=(\mathcal{A}=(V,E),1,\{F_j~\mid ~ 1\leq j\leq n )$,
where for $i<n$ and $1\leq j\leq n$.
\begin{align*}
V&=\{i,x_i,\neg x_i\mid 1\leq i\leq n\} \cup\{\bot\} &
Succ(\bot)&=\{\bot\} & Succ(x_n)=Succ(\neg x_n)&=\{\bot\}\\
Succ(i)&=\{x_i,\neg x_i\} & Succ(x_i)&=Succ(\neg x_i)=\{i+1\} & F_j=c_j
\end{align*}

The partition of $V$ into $(V_\Ev,V_\Ad)$ is such that $i \in V_\Ad$ if $x_i$ is universally quantified, while $i \in V_\Ev$ if $x_i$ is existentially quantified. (The remaining vertices have a single successor, so that it does not matter if they are Adam's or Eve's vertices.)

\begin{thm}\cite{fij-horn2010}
Player $\Ev$ wins the game $G_\phi$ iff $\phi$ is true.
\end{thm}

\begin{expl}
For an example of the reduction, consider the following QBF formula:
\[\phi_1=
\forall x. \, \exists y. \, \forall z. \, \exists u. \, \big( (\neg x \lor \neg y\lor u) \land (x \lor \neg z) \land (\neg z \lor y) \big).
\]
The reduced generalised reachability game $G_{\phi_1}$ is shown in Figure~\ref{fig:qbf}. %
$\Ev$ takes care of existential quantification in $\phi_1$, and $\Ad$ of universal quantification. The clauses correspond to the sets \( F_i \); in this example, we have 
\[
F_1 = \{\neg x, \neg y, u\}, \quad F_2 = \{x, \neg z\}, \quad F_3 = \{\neg z, y\}.
\]

\begin{figure}[h!]\label{fig:qbf}
 \centering
\scalebox{0.8}{
\begin{tikzpicture}[shorten >=0.7pt, node distance=1.6cm, on grid, auto, thick, >=stealth']
   \tikzstyle{state}=[circle, draw, fill=gray!10, minimum size=8mm]
   \tikzstyle{box}=[rectangle, draw, fill=gray!10, minimum size=8mm]

   \node[box] (start) {};
   \node[state] (x) [above right=of start, xshift=1cm] {$x$};
   \node[state] (notx) [below right=of start, xshift=1cm] {$\neg x$};

   \node[state] (mid1) [right=of start, xshift=2.5cm] {};

   \node[state] (y) [above right=of mid1, xshift=1cm] {$y$};
   \node[state] (noty) [below right=of mid1, xshift=1cm] {$\neg y$};

   \node[box] (mid2) [right=of mid1, xshift=2.5cm] {};

   \node[state] (z) [above right=of mid2, xshift=1cm] {$z$};
   \node[state] (notz) [below right=of mid2, xshift=1cm] {$\neg z$};

   \node[state] (mid3) [right=of mid2, xshift=2.5cm] {};

   \node[state] (u) [above right=of mid3, xshift=1cm] {$u$};
   \node[state] (notu) [below right=of mid3, xshift=1cm] {$\neg u$};

   \node[state] (end) [right=of mid3, xshift=2.5cm] {};

   \path[->]
   (start) edge (x)
   (start) edge (notx)
   (x) edge (mid1)
   (notx) edge (mid1)
   (mid1) edge (y)
   (mid1) edge (noty)
   (y) edge (mid2)
   (noty) edge (mid2)
   (mid2) edge (z)
   (mid2) edge (notz)
   (z) edge (mid3)
   (notz) edge (mid3)
   (mid3) edge (u)
   (mid3) edge (notu)
   (u) edge (end)
   (notu) edge (end)
   (end) edge [loop right] (end);
\end{tikzpicture}
}
 \caption{The generalised reachability game for the formula $\phi_1=
\forall x. \, \exists y. \, \forall z. \, \exists u. \, \big( (\neg x \lor \neg y\lor u) \land (x \lor \neg z) \land (\neg z \lor y) \big)$.}
    
\end{figure}
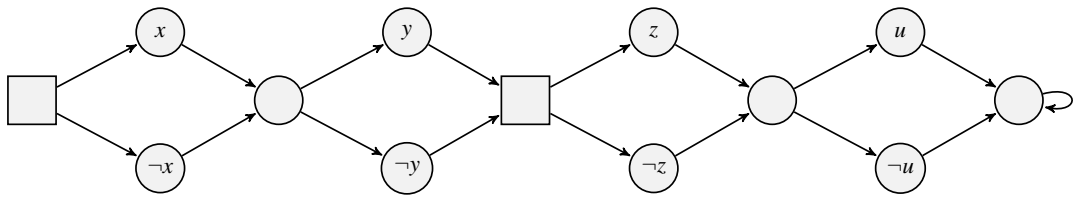

\end{expl}

\subsection{Generalised Reachability with All but One Targets Singleton}
Next we show that the solution of generalised reachability games becomes tractable
when the sizes of the individual targets sets are sufficiently restricted.
We first consider the version with total $k+1$ target sets of which $k$ are singleton targets $F_1=\{t_1\}, \dots, F_k=\{t_k\}$, and one target set $F_0$ has more than one element (that is, $|F_0| > 1$). 
Let $T = \{t_1,\dots,t_k\}$ and $\F = \{F_0\}$ and consider generalised reachability games $\mathcal{G} = (\mathcal{A}, s, \F, T)$.

\begin{thm}\label{thm:genReachSingletons}
Let $\game = (\Aa,s,\F=\{F_0\},T=\{t_1,\dots,t_k\})$ be a game. Then, $\genreach$ is in $\P$.
\end{thm}

\begin{proof}

Let $\mathcal{G} = (\mathcal{A}, s, \F, T)$ be a generalised reachability game
in which all target sets except $F_0$ are singleton sets.

Let $A_1,\ldots,A_k$ denote the attractor sets to the singleton targets sets
$T=(t_1,\ldots,t_k)$, that is let $A_i=\AtE(\{t_i\})$ for $1\leq i \leq k$.
We point out that $t_i\in A_i$.
We claim that $\Ev$ wins the game $\mathcal{G}$ from $s$ if and only if
\begin{enumerate}
\item the sets $A_i$ form a total preorder under set inclusion ; we will assume w.l.o.g.\ that $A_i \subseteq A_{i-1}$ holds;
\item the minimal attractor set contains $s$; and 
\item there is some $0\leq i \leq k$ such that $t_{i+1}\in \AtE(A_{i} \cap F_0)$.
\end{enumerate}

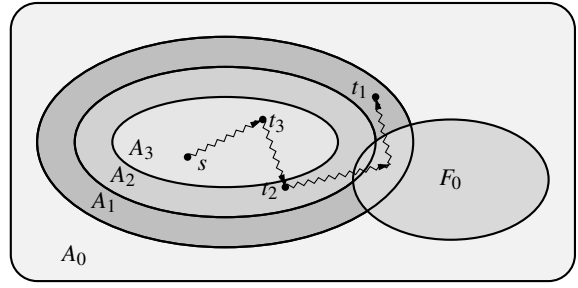
\begin{figure}[h!]
 \centering

\begin{tikzpicture}[>={[inset=0,angle'=27]Stealth}, SharpSquiggly/.style={
            decorate,
            decoration={zigzag, segment length=4, amplitude=0.9},
        }]
\draw[thick,fill=gray!10,rounded corners=10,shift={(0.15,0.15)}]
    (-3,-2) rectangle (4.5,1.7);
\draw [thick,fill=gray!50](0,0) ellipse (2.5cm and 1.4cm);
\draw [thick,fill=gray!35](0,0) ellipse (2cm and 1cm);
\draw [thick,fill=gray!20](0,0) ellipse (1.5cm and .6cm);
\draw [thick,fill=gray!30](3,-0.5) ellipse (1.3cm and 0.8cm);
\draw [thick](0,0) ellipse (2cm and 1cm);
\draw [thick](0,0) ellipse (2.5cm and 1.4cm);
    
\node at (-1.1, -.1){\footnotesize{$A_3$}};
\node at (-1.35, -0.45){\footnotesize{$A_{2}$}};
\node at (-1.6, -0.8){\footnotesize{$A_{1}$}};
\node at (-2, -1.5){\footnotesize{$A_{0}$}};
\node at (3, -.5){\footnotesize{$F_{0}$}};
\node at (-.5,-.2) [circle,fill,inner sep=1pt]{};
\node at (-.3,-.3){\footnotesize{$s$}};
\node at (.5,.3) [circle,fill,inner sep=1pt]{};
\draw[->]  (-.5,-.2) decorate[SharpSquiggly]{ -- (.5,.3) } ;
\node at (.7,.25){\footnotesize{$t_3$}};
\draw[->]  (.5,.3) decorate[SharpSquiggly]{ -- (.8,-.6) } ;
\node at (.8,-.6) [circle,fill,inner sep=1pt]{};
\node at (.6,-.65){\footnotesize{$t_2$}};
\draw[->]  (.8,-.6) decorate[SharpSquiggly]{ -- (2.2,-.3) } ;
\node at (2,.6) [circle,fill,inner sep=1pt]{};
\node at (1.8,.7){\footnotesize{$t_1$}};
\draw[->]  (2.2,-.3) decorate[SharpSquiggly]{ -- (2,.6) } ;
\end{tikzpicture}
\caption{Example construction of winning strategy from attractor sets, $k=4$}
\label{fig:genreach}
\end{figure}

The attractor computations and the check whether they form a total preorder can be implemented in polynomial time. The same holds
for the check whether the minimal attractor set contains $s$ and whether one of the target states $t_{i+1}$ is contained in the attractor to $A_i\cap F_0$. 
Hence the Theorem follows from the claim.

For one direction of the claim, suppose that the attractors are indeed comparable, that is, that for every pair of attractors $A_i$ and $A_j$, we have either $A_i \subseteq A_j$ or $A_j \subseteq A_i$;
also assume that there is some $i \in \{0,\ldots,k\}$ such that $t_{i+1}$ is contained in the attractor to $A_{i} \cap F_0$.
As the order in which the individual targets in a generalised reachabiliy objective are visited is irrelevant, we can reorder the target sets. Without loss of generality, we assume 
\ $A_k\subseteq \ldots \subseteq A_1$ for simplicity. Let the initial state in $\mathcal{G}$ be denoted by $t_{k+1}=s$ and assume $t_{k+1}\in A_k$.
Additionally, let $A_0$ denote the set of all states.
Figure~\ref{fig:genreach} shows an example of this arrangement, where $s=t_4$ and we assume that $i=1$, that is, that $t_2$ is contained in $\AtE(A_1\cap F_0)$.

For $\Ev$ to win the generalised reachability objective from $s$, it is necessary that $t_{j+1} \in A_j$ 
for each $j \in \{1,\ldots,k\}$. This however follows from $A_k\subseteq \ldots \subseteq A_1$, which ensures that each subsequent target state belongs to the corresponding attractor, enabling visits to all target states.

In more detail, we construct a winning strategy as follows:

1. For each $j \in \{1,\ldots, k\}$ with $j\neq i+1$: move from $t_{j+1}$ to $t_j$. Since $t_{j+1}$ is in the attractor $A_j$, player $\Ev$ can enforce that $t_j$ is eventually reached.

2. For the state $t_{i+1}$: move to a state from the set $A_{i} \cap F_0$; this is possible since $t_{i+1}\in \AtE(A_{i} \cap F_0)$. If $i = 0$,  then the strategy terminates successfully. Otherwise, proceed by moving from the state from $A_{i} \cap F_0$ on to $t_{i}$ (which is possible since that state is contained in $A_{i}$), and finish the sequence from $t_i$.

This strategy ensures that each move remains within the attractors, ultimately leading to satisfaction of the objective by visiting all target sets, including $F_0$.

For the converse direction, assume that the attractors do not form a total preorder, that $s$ is not contained in the minimal attractor set, 
or that we have $t_{i+1}\notin \AtE(A_{i} \cap F_0)$ for all $i$.

If the attractors do not form a total preorder, then there exist indices 
$i,j$ such that neither $A_i \subseteq A_j$ nor $A_j \subseteq A_i$ holds.

Define a strategy for player $\Ad$ as follows. Play arbitrarily until $t_i$ or $t_j$ is reached. Assume without loss of generality that we arrive at $t_i$ first. 
Since $t_i \notin A_j$ (otherwise, it would imply $A_i \subseteq A_j$, contradicting our assumption that the attractors do not form a total preorder), player $\Ev$ cannot attract to $t_j$
from $t_i$. From this point onward, the strategy for player $\Ad$ is to remain outside of $A_j$. This straightforward safety strategy ensures that any play following the strategy never enters $A_j$, and thus, will never reach $t_j$. As a result, the reachability player $\Ev$ fails to achieve her goal.

We consider the remaining case that the attractors form a total preorder, but that $s$ is not contained in the minimal attractor set, or that we have $t_{i+1}\notin \AtE(A_{i} \cap F_0)$ for all $i$.
In the former case, player $\Ev$ cannot attract from $s$ to $t_k$. Thus player $\Ad$ wins from $s$ by simply avoiding $t_k$ forever.
In the latter case, player $\Ad$ wins by staying within the attractors $A_i$ but avoiding the sets $A_{i} \cap F_0$ for all $i$. While this strategy may visit all singleton targets $t_i$,
it avoids the set $F_0$ forever so that again $\Ad$ wins.

\end{proof}

\subsection{Generalised Reachability with Mostly Singleton Targets}

Next we consider the case of generalised reachabiliy games with a total of $n+k$ target sets: $n$ singleton targets $t_1, \dots, t_n$ and $k$ target sets $F_1,\dots,F_k$ with $|F_i| > 1$. Let $T = \{t_1,\dots,t_n\}$ and $\F = \{F_1,\dots,F_k\}$ and consider generalised reachability games $\mathcal{G} = (\mathcal{A}, s, \F, T)$.

\begin{thm}Let $\game=(\mathcal{A},s,\F,T)$ be a game with $m$ edges in $\mathcal{A}$, $|T|=n$ and $|\F|=k$.  
Then, $\genreach$ can be solved in time $\mathcal{O}(mn2^k)$.
\end{thm}
\begin{proof}
Let $\mathcal{G} = (\mathcal{A}, s, \F, T)$ be a generalised reachability game with parameters as stated in the claim.
First, we observe that a necessary condition for player $\Ev$ to win $G$ is that 
$\Ev$ wins $(\mathcal{A}, s, \emptyset, T)$ as well~\cite{fij-horn2010}.
If player $\Ev$ wins $(\mathcal{A}, s, \F, T)$,
we hence can follow the proof of Theorem~\ref{thm:genReachSingletons} and
assume without loss of generality a total order on $T$ as $t_1,\dots,t_n$ where $s = t_0 \in \AtE{(t_1)}$ and $t_i \in \AtE{(t_{i+1})}$. Then let $T'$ denote $\{t_0\}\cup T$.

Next we transform the game arena of $\mathcal{G}$ in order to treat to non-singleton target sets $\mathcal{F}$. To this end, 
let $V$ denote the set of game nodes of $\mathcal{A}$ and 
consider the game arena $\hat{\mathcal{A}}$ over $\hat{V} = V \times 2^{[k]}$, that is, $V$ is augmented with the memory structure $2^{[k]}$, storing the set of target sets from
$\mathcal{F}$ that have been satisfied so far. States $(u,S)\in\hat{V}$ are owned by the player owning $u$ in $\mathcal{A}$ and we have an edge $(u,S) \rightarrow (v,S')$ in $\hat{\mathcal{A}}$ iff (i) $u \rightarrow v$ is an edge in $\mathcal{A}$ and (ii) $S' \subseteq S \cup \{j\in\{1,\ldots,k\} \mid v \in F_j\}$.

Next, we inductively define a sequence of sets $D_i$ by putting $D_{n+1} = V \times \{1,\ldots,k\}$ and, for $0\leq i\leq n$,
$D_i = \{(u,S) \in A_{i+1} \mid u = t_{i}\}$, where $A_{i+1} = \AtE(D_{i+1})$.

We claim that $\Ev$ wins ${\mathcal{G}}$ if and only if $(s,\emptyset)\in D_0$.

For one direction, let $(s,\emptyset)\in D_0$. Then $\Ev$ can just follow her individual attractor strategies one after another, which results in visits
to all target states $t_i\in T$, as well to all target sets $F_i\in \F$; the latter is the case since the values of the auxiliary memory indicate the 
target sets that have been visited so far, and since the constructed strategy ensures that a game node from
$D_{n+1}$, that is, with auxiliary memory value $\{1,\ldots,k\}$ is eventually reached.

For the converse direction, let $(s,\emptyset)\notin D_0$.
We show that $\Ad$ wins $\mathcal{G}$.
We first assume that the attractors of all target sets are different.
Then we construct a winning strategy for $\Ad$ in $\mathcal{G}$ as follows.
We note that $\Ad$ can follow the main objective to force the order in which the individual $t_i$ are reached.
For this, he first and foremost wins if the first component of $\hat{\mathcal{A}}$ (the original game node in $\mathcal{A}$) falls outside of the attractor for the next target
state.

With this in mind, $\Ad$ can just follow a strategy in $\mathcal{G}$ that attempts to prevent that the induced play on $\hat{\mathcal{A}}$ ever reaches $D_1$ from states of the form $(t_0,S) \notin D_0$; if this fails, $\Ad$ uses his strategy to prevent states of the form $(t_1,S) \notin D_1$ to reach $D_2$, and so on.
Since $(s,\emptyset)\notin D_0$, $\Ad$ can use this strategy to ensure
that whenever a game node $u$ is reached by a play in $\mathcal{G}$ that visits all target sets $F_i\in \mathcal{F}$ (which corresponds to reaching the node $(u,\{1,\ldots,k\})\in D_{n+1}$ in
the induced play on $\mathcal{A}'$), then at least one of the target states $t_i$ has not been visited by the play so far. Thus the described strategy indeed is winning for $\Ad$, as required.

We note that if two or more target states have the same attractor, then they are neighbours $t_i,\ldots,t_j$, and $D_i,\ldots,D_j$ only differ by their first component.
For deciding reachability, we can then just keep one such state, keeping in mind that we can reach all other such states and return to any of these target states afterwards.

Regarding runtime complexity, the proposed solution algorithm computes $n+2$ sets $D_i$ and $A_i$; both are subsets of $V\times 2^{[k]}$. Given $A_{i+1}$, the computation of a single set $D_i$ can be done in time linear in $m \cdot 2^k$.
The computation of a single attractor set $A_i$ over a graph with $m \cdot 2^k$ edges takes time $\mathcal{O}(m \cdot 2^k)$. Hence the overall algorithm can be implemented to run in time
$\mathcal{O}(nm2^k)$.
\end{proof}

This in particular provides fixed parameter tractability.

\begin{col}
$\genreach$ is in \P~when the number of target sets that are not singleton is logarithmic in the size of the game.
\end{col}
\subsection{One Player Case with $\Ad$}
Here we consider the special case of generalised reachability games with just one player, $\Ad$.
For one-player games with $\Ev$, $\genreach$ is known to be $\NP$-complete in general, and in $\P$ when $|F_i| \leq 2$ for all $F_i \in \F$~\cite{fij-horn2010}.
If $|F_i|=1$ for all $F_i\in \F$, \cite[Theorem 5]{austin2025} can be modified to obtain $\NL$-completeness. 
On the other hand, for the one-player case with $\Ad$, the general case was shown to be in $\P$~\cite{fij-horn2010}. We provide an improved complexity bound by showing that the single-player case with $\Ad$ is in fact $\NL$-complete. 
\begin{thm}
$\genreach$ is $\NL$-complete when all vertices in $V$ belong to $\Ad$. The $\NL$-hardness holds even when $|\F|=1$. %
\end{thm}
\begin{proof}
First, we show containment in $\NL$. $\Ad$ wins if he has a strategy to ensure that some $F_i \in \F$ is never visited. In other words, $\Ad$ wins from vertex $v$ iff $\exists F_i \in \F$ such that $v \not \in \AtE(F_i)$. %
Observe that, if $\Ad$ wins, there is always a winning strategy of $\Ad$ that produces a \emph{lasso} play, i.e. a play of the form $v_0 \dots v_i \dots v_m$ where $m \leq |V|$, $v_i = v_m$ and no vertex is repeated except $v_i = v_m$. This is because, if a winning strategy produces a play with more than one loop, then $\Ad$ has another winning strategy where at least one of those loops is not taken. We call such a play an $(v_0,v_i,v_i)$ lasso.  
Based on this, we will provide an $\NL$ algorithm for the complement decision problem i.e. for checking if $\Ad$ has a winning strategy.
Since $\NL = \coNL$, this gives an $\NL$ upper bound. 

We non-deterministically guess a $(s,t,t)$ lasso and a set $F_i$ such that no vertex from the lasso is in $F_i$.
We guess the lasso by guessing successors step by step starting from $s$. At each step we check that the vertex is not in $F_i$. We also non-deterministically guess the $t$ at some step and store it. We also maintain a counter to store the length of the play so far. The algorithm terminates when $t$ is repeated and the length is no more than $n$. This algorithm uses logarithmic space.

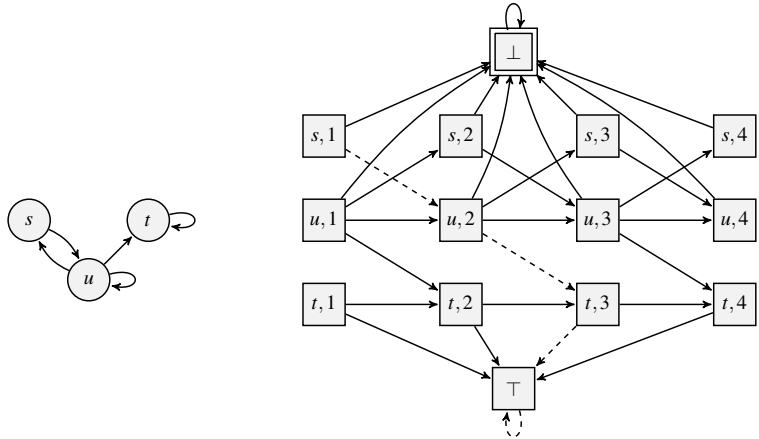
\begin{figure}[h!]
 \centering
\scalebox{0.7}{
\begin{tikzpicture}[shorten >=0.7pt, node distance=1.6cm, on grid, auto, thick, >=stealth']
   \tikzstyle{state}=[circle, draw, fill=gray!10, minimum size=8mm]
   \tikzstyle{box}=[rectangle, draw, fill=gray!10, minimum size=8mm]

   \node[state] (s) [yshift=-0.8cm] {$s$};
   \node[state] (u) [below right=of s] {$u$};
   \node[state] (t) [above right=of u] {$t$};

   \node[box] (s1) [right=of start, xshift=4cm,yshift=0.8cm] {$s,1$};
   \node[box] (u1) [below=of s1] {$u,1$};
   \node[box] (t1) [below=of u1] {$t,1$};

   \node[box] (s2) [right=of s1, xshift=1cm] {$s,2$};
   \node[box] (u2) [below=of s2] {$u,2$};
   \node[box] (t2) [below=of u2] {$t,2$};

   \node[box] (s3) [right=of s2, xshift=1cm] {$s,3$};
   \node[box] (u3) [below=of s3] {$u,3$};
   \node[box] (t3) [below=of u3] {$t,3$};

   \node[box] (s4) [right=of s3, xshift=1cm] {$s,4$};
   \node[box] (u4) [below=of s4] {$u,4$};
   \node[box] (t4) [below=of u4] {$t,4$};

   \node[box,double, double distance=0.7mm] (bot) [above=of s2, xshift=1cm] {$\bot$};
   \node[box] (top) [below=of t2, xshift=1cm] {$\top$};

   \path[->]

   (s) edge [bend left=20] (u)
   (u) edge [bend left=20] (s)
   (u) edge (t)
   (u) edge [loop right] (u)
   (t) edge [loop right] (t)

   (s1) edge [dashed] (u2)
   (u1) edge (u2)
   (u1) edge (s2)
   (u1) edge (t2)
   (t1) edge (t2)
   
   (s2) edge (u3)
   (u2) edge  (u3)
   (u2) edge (s3)
   (u2) edge [dashed] (t3)
   (t2) edge (t3)

   (s3) edge (u4)
   (u3) edge (u4)
   (u3) edge (s4)
   (u3) edge (t4)
   (t3) edge (t4)

   (t1) edge (top)
   (t2) edge (top)
   (t3) edge [dashed] (top)
   (t4) edge (top)

   (s1) edge (bot)
   (s2) edge (bot)
   (s3) edge (bot)
   (s4) edge (bot)
   (u1) edge [bend left=10] (bot)
   (u2) edge [bend right=10] (bot)
   (u3) edge [bend left=10] (bot)
   (u4) edge [bend right=10] (bot)

   (top) edge [dashed, loop below] (top)
   (bot) edge [loop above] (bot);
\end{tikzpicture}
}
 \caption{Example of reduction from $s-t$ reachability to one-player generalised reachability.}
    \label{fig:nl}
\end{figure}

For the lower bound, we provide a reduction from the $\NL$-complete $s-t$ reachability problem. The $s-t$ reachability problem asks whether a vertex $t$ is reachable from a vertex $s$ in a given graph $H = (V_H,E_H)$ with $s,t \in V_H$. We construct a game $\Gg$ such that $t$ is reachable from $s$ in $H$ if and only if $\Ad$ has a winning strategy.

Let $|V_H| = n$. The game $\Gg$ has an arena $\Aa$ with underlying graph $G = (V,E)$.  We have $V = |V_H| \times [n+1] \cup \{ \top, \bot\}$, all of which are owned by $\Ad$. If $(u,v) \in E_H$, then for every $1\leq i \leq n$, in $E$ we have $(u,i) \rightarrow (v,i+1)$. For every $i$ we have $(t,i) \rightarrow \top$. For every $v \in V\setminus \{t\}$ and for each $i$, we have $(v,i) \rightarrow \bot$. The vertices $\top$ and $\bot$ have self loops. $(s,1)$ is the start vertex and $\Gg$ has only one singleton target set $T = \{\bot\}$. %
We note that $\Gg$ can be constructed in logspace from $H$.
Figure~\ref{fig:nl} shows an example of this reduction for a graph with three vertices; the dashed arrows indicate a strategy for $\Ad$ to spoil reachability of $\{\bot\}$ in the reduced game,
corresponding to the fact that $t$ is reachable from $s$ in the graph.

We claim that in general, $t$ is reachable from $s$ in $H$ iff $\Ad$ wins the game $\Gg = (\Aa,(s,1),\{\bot\})$. %
To prove the claim, we observe that all infinite paths in $G$ end up in exactly one of the self loops at $\top$ or $\bot$. If $t$ is reachable from $s$, then $\Ad$ has a path from $(s,1)$ to $(t,i)$ for some $i$ and then can move to $\top$ and hence wins the game, i.e. $\Ad$ has a winning $((s,1),\top,\top)$ lasso. If $t$ is not reachable from $s$, then all paths starting from $(s,1)$ end up at $\bot$, in which case $\Ad$ loses. 
\end{proof}

\begin{remark}
For one-player games with $\Ad$, even for the optimisation variants of the problem, i.e. the $\maxgenreach$ and $\maxgenreachpromise$ problems, $\Ad$ always has an optimal strategy that produces a lasso. This follows from the fact that it is always better for $\Ad$ to encounter fewer distinct vertices, and consequently fewer target sets. 
\end{remark}

\section{Maximum Generalised Reachability}

In this section, we address the complexity of the $\maxgenreach$ and the $\maxgenreachpromise$ problems. %
\subsection{One Player Case with $\Ev$}
We first consider the case where all vertices of the game are controlled by $\Ev$. In this case, a strategy of $\Ev$ produces a unique play $\rho$. Therefore, the $\maxgenreach$ and $\maxgenreachpromise$ problems are equivalent if $\Ev$ pledges the set of targets seen in the play $\rho$. Hence we only state our results for the $\maxgenreach$ variant.

\begin{thm}
\label{thm:maxgrev}
Let $\game$ be a game where $\Ev$ controls all the vertices, i.e, $V=V_{\Ev}$.
\begin{enumerate}
    \item $\maxgenreach$ is in \P~when each target set $F_i \in \F$ is of size $1$.
    \item $\maxgenreach$ is \NP-complete in general. It is \NP-hard even when $|F_i| = 2$ for each target set $F_i \in \F$. 
\end{enumerate}
\end{thm}

\begin{proof}
For the case with all target sets of size $1$, we present an algorithm that runs in polynomial time. The algorithm first computes the strongly connected component (SCC) decomposition of the arena. For each SCC, we assign a value equal to the number of target vertices contained in it. Within each SCC, $\Ev$ can visit all the target vertices. Computing the maximum number of target vertices that $\Ev$ can visit corresponds finding a path in the SCC decomposition that maximises the sum of the value of the SCCs contained in the path.
This can be computed bottom-up using dynamic programming starting from the bottom SCCs. 
This is explained in detail in the proof of Theorem \ref{thm:maxgrp2p}, as the algorithm also works for the two-player case.
When target sets are of size $2$, NP-hardness follows from MAX-2-SAT problem, which is known to be NP-hard~\cite{Halaby2016}. This follows the same reduction as the one used to show PSPACE-hardness in the general case with 2 players and arbitrary target sets~\cite{fij-horn2010}, presented in Figure~\ref{fig:qbf}.

To see membership in NP, $\Ev$ can guess a path of size at most $nk$ in length that visits $k$ target sets, where $n$ is the number of vertices in the arena. This is possible as the length of the path between any two consecutive target vertices not seen before is at most $n$.
\end{proof}

\subsection{One Player Case with $\Ad$}

Recall that the $\maxgenreach$ and $\maxgenreachpromise$ problems can have different solutions even when $\Ad$ is the sole player, as demonstrated by a game obtained by modifying Figure~\ref{fig:maxproblem}, where all vertices belong to Adam : in this case, Eve cannot promise to visit any \emph{particular} target vertex, but she can ensure that at least one target vertex is visited.
In this section we deal with the one-player variant of these problems with $\Ad$ as the only player. First we provide the full complexity picture for the $\maxgenreach$ problem. 
\begin{thm}
\label{thm:maxgrad}
Let $\game$ be a game where $\Ad$ controls all the vertices, i.e, $V=V_{\Ad}$. 
\begin{enumerate}
    \item  $\maxgenreach$ is in $\P$ when all target sets are singleton.
    \item $\maxgenreach$ is $\coNP$-complete in general. The $\coNP$-hardness holds even when $|F_i| = 2$ for each $F_i \in \F$. 
\end{enumerate}
\end{thm}
\begin{proof}
We begin with a polynomial time algorithm for the case where target sets have size $1$. Note that a strategy of $\Ad$ is just a path $\rho$. We can assume that $\rho$ is a path of the form $s\to t\to t$. Otherwise, one can find a minimal loop in $\rho$ and simply repeat it to obtain a path $\rho'$ of the correct form. Since $\rho'$ visits a subset of vertices visited by $\rho$, $\rho'$ should be at least as good as $\rho$ for $\Ad$. Therefore, we search for paths of the form $s\to t\to t$ which visit the least number of target vertices.

We assign weights to edges of the graph $G$ using the function $w: E\to \{0,1\}$ as follows: 
$w(u,v)=1$ if, and only if, $\{v\}\in \F$.
The problem reduces to computing the lightest weighted $s\to t\to t$ path in the weighted graph. 
This can be done by iterating over all choices of $t$ and computing the lightest such path.

For co-NP hardness, we look at the complement problem, i.e. given a game $\game$ with $V=V_\Ad$ and $k$, is there a strategy of $\Ad$ such that on playing this strategy at most $k$ targets are visited. We consider the problem MIN-2-SAT, which is known to be NP-hard \cite{Kohli-et-al-94}. 
Following the same reduction as in  NP-hardness in proof of Theorem \ref{thm:maxgrev}, but giving control of the vertices to $\Ad$, we get a game in which $\Ad$ can visit at most $k$ target vertices if and only if there is an assignment which satisfies at most $k$ clauses.  

For the upper bound, we observe that $\Ad$'s strategy can be simplified to a path of length at most $n+1$, because once he sees a vertex twice, he can simply repeat the loop without visiting more target sets. We can thus guess a path of length $n+1$ and check if this hits at most $k$ targets. This puts the complement problem in NP and thus shows the co-NP membership. 
\end{proof}

Note that the algorithm to solve $\genreach$ games where all vertices are controlled by $\Ad$ computes $\AtE$ for each target set, and checks if the initial vertex is in the $\AtE$ for all target sets or not. If not, $\Ad$ has a choice to avoid a target set and win. This runs in polynomial time as computing the 
$\AtE$ can be done in polynomial time. However, this does not work for the $\maxgenreach$ problem as the initial state $s$ might not be in the $\AtE(F_i)$ and $\AtE(F_j)$, but in the $\AtE(F_i\cup F_j)$. Thus, the play will visit at least $1$ of the target sets.
For example, in the game from Figure~\ref{fig:maxproblem}, $s$ and $u$ are neither in $\AtE(\{u_1\})$, nor in $\AtE(\{u_3\})$, but they are in $\AtE(\{u_1,u_3\})$.

However, the polynomial time algorithm can be adapted to work for the $\maxgenreachpromise$ problem.

\begin{thm}
\label{thm:maxgrprad}
Let $\game$ be a game where $\Ad$ controls all the vertices, i.e, $V=V_{\Ad}$. Then,~\\   
  $\maxgenreachpromise$  is in \P.
\end{thm}
\begin{proof}
    The $\maxgenreachpromise$ problems asks, if there is a $k$ sized subset of target sets such that the play visits a target from each of the $k$ chosen target sets. 
    The algorithm computes the $\AtE(F_i)$ for all target sets $F_i$ and counts how many of them contain the initial state $s_0$. If $s_0$ is contained in at least $k$ of the $\AtE$ sets, then any play will visit the $k$ target sets. This is because being in the attractor set of a target $F_i$ guarantees that Eve has a strategy to force the play to reach $F_i$, even though Adam controls all the vertices. Therefore, $\Ev$ can promise $k$ target sets whose $\AtE$ contains $s_0$. If $s_0$ is not contained in $k$ attractor sets, then no matter which $k$ sets $\Ev$ promises, there will be a promised target set $F'$, such that $s_0\not \in \AtE(F)$. Therefore, $\Ad$ will be able to avoid the promised target sets.
\end{proof}

\subsection{Two Player Case}
Here we discuss the complexity of the two-player case.
We first state the results for the $\maxgenreach$ problem. %

\begin{thm}
Let $\game$ be a game.
\begin{enumerate}
\item $\maxgenreach$ is \PSPACE-complete. The \PSPACE-hardness holds even when $|F_i| = 2$ for each $F_i \in \F$.
\item $\maxgenreach$ is NP-hard when each $F_i \in \F$ is singleton.

\end{enumerate}
\end{thm}

\begin{proof}
(1)     
The upper bound in the general case, i.e, $|F_i|\ge 2$, for all $F_i\in \F$, follows from \cite[Theorem 12]{KupfermanS24}. In fact, they consider the question of maximizing weighted reachability, where different reachability objectives can have different weight associated to them and $\Ev$ attempts to maximise the total weight. Our case corresponds to the case where all the reachability objectives have the same weight, which is called the \textit{MaxR} objective in \cite{KupfermanS24}.

The PSPACE lower bound presented in \cite{KupfermanS24} follows from the PSPACE-hardness of solving generalised reachability game \cite{fij-horn2010}, which works for target sets of size $\ge 3$. We present a hardness proof for the case where target sets are of size $2$.
In the reduction from 3-SAT to MAX-2-SAT~\cite{GAREY1976237}, given a set of 3-CNF clauses, they construct a set of 2-CNF clauses such that exactly $\frac{7}{10}$ of the 2-CNF clauses are satisfied if the original 3-CNF clause was satisfied. Otherwise, at most $\frac{6}{10}$ of the 2-CNF clauses are satisfied. Since this holds for any assignment, we can transform the matrix of the QBF formula from a 3-CNF to a 2-CNF formula such that exactly $\frac{7}{10}$ of the clauses are satisfied by a satisfying assignment. Therefore, this proves that MAX-2-QSAT is PSPACE-hard. Using the construction to translate QSAT to generalised reachability games, we obtain that $\maxgenreach$ with target sets of size $2$ is PSPACE-hard.

(2) We reduce from the minimum vertex cover problem, which is known to be NP-complete~\cite{Karp72}.
Given an undirected graph $G = (V,E)$ and a number $\ell$, the minimum vertex cover problem asks whether $G$ has a vertex cover $S \subseteq V$ of size at most $\ell$. We construct a game as follows: the vertices controlled by $\Ev$, $V_\Ev=V$, correspond to the vertices of the original graph $G$ and the vertices controlled by $\Ad$, $V_\Ad=E$, correspond to the edges of the original graph $G$. 
From any vertex in $V_\Ev$, $\Ev$ can choose any edge $(u,v)$ of the original graph $G$ and move to the vertex corresponding $(u,v)\in V_\Ad$. From a vertex $(u,v)$ in $V_\Ad$, $\Ad$ can move to vertices $u,v\in V_\Ev$, corresponding to the two end points of the edge.
The start vertex is any arbitrary vertex in $V_\Ad$. 
The target set $\F=\{\{v\}~\mid~ v\in V_\Ev\}$ consists of singleton sets containing each vertex of $V$.
We claim that $G$ has a vertex cover of size at most $\ell$ iff $\Ev$ can reach at least $\ell$ targets. 

Consider the following strategy of $\Ev$: assuming any ordering on the set $E$, $\Ev$ chooses the next edge in this order in a round-robin manner. The best response strategy of $\Ad$ against this is to play a minimum vertex cover of $G$.
Also against this strategy of $\Ad$, $\Ev$ cannot do any better since the number of vertices visited in $V_\Ev$ is at most the size of the minimum vertex cover. Hence, the value $k$ of this game is exactly the size of a minimum vertex cover.
This proves our claim and completes the proof of NP-hardness.
\end{proof}

Note that the NP-hardness of the $\maxgenreach$ problem with target sets of size $1$ is in contrast with the results for $\genreach$ problem which is known to be in P for target sets of size $1$.

Next, we state our results for the two-player case of $\maxgenreachpromise$.

\begin{thm}
\label{thm:maxgrp2p}
\phantom{a}
\begin{enumerate}
\item The $\maxgenreachpromise$ problem for the two-player case is in $\P$ when all target sets are singleton.
\item $\maxgenreachpromise$ is \PSPACE-complete in general. It is $\PSPACE$-hard even when all target sets have size $3$.
\end{enumerate}
\end{thm}
\begin{proof}

We give a polynomial time algorithm for the case when all target sets are singleton.
The algorithm below is similar to the generalised reachability case and proceeds as follows:

\begin{enumerate}
    \item Compute the attractor relation between target vertices %
    forming a preorder graph capturing these relationships. Formally, with $t_0 = s$, consider the relation $t_i \preceq t_j$ iff $t_i \in \AtE(t_j)$, for $i\in \{0\}\cup [n]$ and create a graph with $t_i$ as vertices and edges from $t_i$ to $t_j$ iff $t_i \preceq t_j$. Call this the preorder graph.
    \item Perform the Strongly Connected Component (SCC) decomposition of this preorder graph.
    \item Assign a weight to each SCC equal to the number of target vertices contained within it.
    \item The resulting SCC decomposition is a Directed Acyclic Graph (DAG). Using bottom-up dynamic programming, find a path in this DAG with the maximum total %
    weight starting from the SCC containing $t_0$.
\end{enumerate}

The set of targets on this path from $t_0$ can be visited because, by construction, the $\Ev$-attractors of the targets are ordered by inclusion.

At the same time, any target set $\Ev$ can promise must have target states whose $\Ev$ attractors are ordered by inclusion; they are therefore on a path from $t_0$ in this DAG.
Thus, there cannot be a larger such set, as it would lead to a path with a larger weight.

This shows that our dynamic programming algorithm works. 
For the complexity, we note that each step can be performed efficiently:
     Attractor computation can be done in \( O(|V| + |E|) \). SCC decomposition can also be computed in \( O(|V| + |E|) \) \cite{SHARIR198167}.
     Assigning weights to SCCs and constructing the weighted DAG is linear in the graph size.
     Finally, computing the maximum-weight path in a DAG via bottom-up dynamic programming is also linear in the size of the DAG obtained via SCC decomposition. Since each step is polynomial, the overall algorithm solves the problem $\maxgenreachpromise$ in polynomial time when each target set is singleton.

\bigskip

For the general case with target sets of arbitrary size, we obtain a \NPSPACE~algorithm by simply guessing the target sets promised by $\Ev$ and then applying the \PSPACE~algorithm for solving generalised reachability games from \cite{fij-horn2010}. As a consequence of Savitch's Theorem, this gives a \PSPACE~upper bound for the general problem. The $\PSPACE$ lower bound holds even for target sets of size $3$ as the generalised reachability game with the objective to visit all targets is a special case of the $\maxgenreachpromise$ problem.
\end{proof}
We leave the exact complexity of $\maxgenreachpromise{}$ for the case with target sets of size $2$ open; Theorem \ref{thm:maxgrev} provides an \NP~lower bound.

\section{Discussion}

We have studied variations of games with generalised reachability objectives and maximum generalised reachability objectives, considering the size of target sets as parameter.
We have provided several complexity results for these two problems, showing first
that generalised reachability can be checked in time linear in (1) the size of the game and (2) the number of singleton target sets, and (3) exponential in the number of larger target sets.
This extends the polynomial time complexity from the case where all target sets are singleton~\cite{fij-horn2010} to allow for a logarithmic number of target sets of arbitrary size.
We have then established NL-completeness for the single-player case, where only the environment makes decisions.

Next, we have introduced optimisation variants of the generalised reachability problem, where the goal generalises from visiting all target sets to visiting as many target sets as possible.
The first natural goal here is to just \emph{maximise} the number of target sets visited; we show that this problem is different in that it is \textsc{NP}-hard, even when each target set contains only a single element, and \textsc{PSPACE}-complete even when the size of target sets is restricted to at most two.
We also show that both single player variants of this problem are tractable if the target sets are singleton, but become intractable already for games with target sets of size two over DAGs with pathwidth two.

We also considered an interesting variant, where $\Ev$ is asked to pledge -- before the game starts -- the target sets that are to be visited.
She then has to visit them all, and her goal shifts to pledging a largest set of target sets.
We show that the solution of games with such objectives is tractable for singleton target sets even in the two-player case. Another interesting variant could require $\Ev$ to specify the order in which the target sets are visited. This variant introduces additional constraints and may not be reducible to the cases considered in this work.

Thus, we clarify the landscape of complexity of several problems in generalised reachability. The major remaining open problem is the complexity of the generalised reachability problem where every target set is of size at most two~\cite{fij-horn2010}. 
The promise variant, $\maxgenreachpromise$ problem with target sets of size $2$, is shown to be NP-hard (even when $\Ev$ is the only player). Since $\genreach$ can be seen as a special case of $\maxgenreachpromise$, this has consequences for the exact complexity of $\genreach$ with target sets of size $2$. In this case, a polynomial time algorithm for $\genreach$ would imply $\maxgenreachpromise$ is harder (unless P=NP). On the other hand, an upper bound better than PSPACE for $\maxgenreachpromise$ would also imply an improved upper bound for $\genreach$ in the case of target sets of size $2$.

Regarding the optimisation variant $\maxgenreach$, the exact complexity of the case with singleton sets remains open as we have shown it to be NP-hard and in PSPACE.

\nocite{}
\bibliographystyle{eptcs} 
\bibliography{references}

\begin{thebibliography}{10}
\providecommand{\bibitemdeclare}[2]{}
\providecommand{\surnamestart}{}
\providecommand{\surnameend}{}
\providecommand{\urlprefix}{Available at }
\providecommand{\url}[1]{\texttt{#1}}
\providecommand{\href}[2]{\texttt{#2}}
\providecommand{\urlalt}[2]{\href{#1}{#2}}
\providecommand{\doi}[1]{doi:\urlalt{https://doi.org/#1}{#1}}
\providecommand{\eprint}[1]{arXiv:\urlalt{https://arxiv.org/abs/#1}{#1}}
\providecommand{\bibinfo}[2]{#2}

\bibitemdeclare{misc}{austin2025}
\bibitem{austin2025}
\bibinfo{author}{Pete \surnamestart Austin\surnameend},
  \bibinfo{author}{Nicolas \surnamestart Mazzocchi\surnameend},
  \bibinfo{author}{Sougata \surnamestart Bose\surnameend} \&
  \bibinfo{author}{Patrick \surnamestart Totzke\surnameend}
  (\bibinfo{year}{2025}): \emph{\bibinfo{title}{Temporal Explorability Games}},
  \doi{10.48550/arXiv.2412.16328}.

\bibitemdeclare{inbook}{Bloem2018}
\bibitem{Bloem2018}
\bibinfo{author}{Roderick \surnamestart Bloem\surnameend},
  \bibinfo{author}{Krishnendu \surnamestart Chatterjee\surnameend} \&
  \bibinfo{author}{Barbara \surnamestart Jobstmann\surnameend}
  (\bibinfo{year}{2018}): \emph{\bibinfo{title}{Graph Games and Reactive
  Synthesis}}, pp. \bibinfo{pages}{921--962}.
\newblock \bibinfo{publisher}{Springer International Publishing},
  \bibinfo{address}{Cham}, \doi{10.1007/978-3-319-10575-8_27}.

\bibitemdeclare{inproceedings}{chatterjee_et_al10}
\bibitem{chatterjee_et_al10}
\bibinfo{author}{Krishnendu \surnamestart Chatterjee\surnameend},
  \bibinfo{author}{Laurent \surnamestart Doyen\surnameend},
  \bibinfo{author}{Thomas~A. \surnamestart Henzinger\surnameend} \&
  \bibinfo{author}{Jean-Fran\c{c}ois \surnamestart Raskin\surnameend}
  (\bibinfo{year}{2010}): \emph{\bibinfo{title}{{Generalized Mean-payoff and
  Energy Games}}}.
\newblock In: {\slshape \bibinfo{booktitle}{Foundations of Software Technology
  and Theoretical Computer Science (FSTTCS 2010)}}, {\slshape
  \bibinfo{series}{LIPIcs}}~\bibinfo{volume}{8}, \bibinfo{publisher}{Schloss
  Dagstuhl -- Leibniz-Zentrum f{\"u}r Informatik}, pp.
  \bibinfo{pages}{505--516}, \doi{10.4230/LIPIcs.FSTTCS.2010.505}.

\bibitemdeclare{inproceedings}{chatterjee_et_al16}
\bibitem{chatterjee_et_al16}
\bibinfo{author}{Krishnendu \surnamestart Chatterjee\surnameend},
  \bibinfo{author}{Wolfgang \surnamestart Dvor\'{a}k\surnameend},
  \bibinfo{author}{Monika \surnamestart Henzinger\surnameend} \&
  \bibinfo{author}{Veronika \surnamestart Loitzenbauer\surnameend}
  (\bibinfo{year}{2016}): \emph{\bibinfo{title}{{Conditionally Optimal
  Algorithms for Generalized B\"{u}chi Games}}}.
\newblock In: {\slshape \bibinfo{booktitle}{Mathematical Foundations of
  Computer Science (MFCS 2016)}}, {\slshape
  \bibinfo{series}{LIPIcs}}~\bibinfo{volume}{58}, \bibinfo{publisher}{Schloss
  Dagstuhl -- Leibniz-Zentrum f{\"u}r Informatik}, pp.
  \bibinfo{pages}{25:1--25:15}, \doi{10.4230/LIPIcs.MFCS.2016.25}.

\bibitemdeclare{inproceedings}{chatterjee2007}
\bibitem{chatterjee2007}
\bibinfo{author}{Krishnendu \surnamestart Chatterjee\surnameend},
  \bibinfo{author}{Thomas~A. \surnamestart Henzinger\surnameend} \&
  \bibinfo{author}{Nir \surnamestart Piterman\surnameend}
  (\bibinfo{year}{2007}): \emph{\bibinfo{title}{Generalized Parity Games}}.
\newblock In: {\slshape \bibinfo{booktitle}{Foundations of Software Science and
  Computational Structures (FoSSaCS 2007)}}, {\slshape \bibinfo{series}{LNCS}}
  \bibinfo{volume}{4423}, \bibinfo{publisher}{Springer}, pp.
  \bibinfo{pages}{153--167}, \doi{10.1007/978-3-540-71389-0\_12}.

\bibitemdeclare{misc}{fijalkow2023gamesgraphs}
\bibitem{fijalkow2023gamesgraphs}
\bibinfo{author}{Nathanaël \surnamestart Fijalkow\surnameend},
  \bibinfo{author}{C.~\surnamestart Aiswarya\surnameend}, \bibinfo{author}{Guy
  \surnamestart Avni\surnameend}, \bibinfo{author}{Nathalie \surnamestart
  Bertrand\surnameend}, \bibinfo{author}{Patricia \surnamestart
  Bouyer\surnameend}, \bibinfo{author}{Romain \surnamestart
  Brenguier\surnameend}, \bibinfo{author}{Arnaud \surnamestart
  Carayol\surnameend}, \bibinfo{author}{Antonio \surnamestart
  Casares\surnameend}, \bibinfo{author}{John \surnamestart
  Fearnley\surnameend}, \bibinfo{author}{Paul \surnamestart Gastin\surnameend},
  \bibinfo{author}{Hugo \surnamestart Gimbert\surnameend},
  \bibinfo{author}{Thomas~A. \surnamestart Henzinger\surnameend},
  \bibinfo{author}{Florian \surnamestart Horn\surnameend},
  \bibinfo{author}{Rasmus \surnamestart Ibsen-Jensen\surnameend},
  \bibinfo{author}{Nicolas \surnamestart Markey\surnameend},
  \bibinfo{author}{Benjamin \surnamestart Monmege\surnameend},
  \bibinfo{author}{Petr \surnamestart Novotný\surnameend},
  \bibinfo{author}{Pierre \surnamestart Ohlmann\surnameend},
  \bibinfo{author}{Mickael \surnamestart Randour\surnameend},
  \bibinfo{author}{Ocan \surnamestart Sankur\surnameend},
  \bibinfo{author}{Sylvain \surnamestart Schmitz\surnameend},
  \bibinfo{author}{Olivier \surnamestart Serre\surnameend},
  \bibinfo{author}{Mateusz \surnamestart Skomra\surnameend},
  \bibinfo{author}{Nathalie \surnamestart Sznajder\surnameend} \&
  \bibinfo{author}{Pierre \surnamestart Vandenhove\surnameend}
  (\bibinfo{year}{2025}): \emph{\bibinfo{title}{Games on Graphs: From Logic and
  Automata to Algorithms}}, \doi{10.48550/arXiv.2305.10546}.

\bibitemdeclare{article}{fij-horn2010}
\bibitem{fij-horn2010}
\bibinfo{author}{Nathanaël \surnamestart Fijalkow\surnameend} \&
  \bibinfo{author}{Florian \surnamestart Horn\surnameend}
  (\bibinfo{year}{2010}): \emph{\bibinfo{title}{The surprizing complexity of
  generalized reachability games}}.
\newblock \doi{10.48550/arXiv.1010.2420}.

\bibitemdeclare{article}{GAREY1976237}
\bibitem{GAREY1976237}
\bibinfo{author}{M.R. \surnamestart Garey\surnameend}, \bibinfo{author}{D.S.
  \surnamestart Johnson\surnameend} \& \bibinfo{author}{L.~\surnamestart
  Stockmeyer\surnameend} (\bibinfo{year}{1976}): \emph{\bibinfo{title}{Some
  simplified NP-complete graph problems}}.
\newblock {\slshape \bibinfo{journal}{Theoretical Computer Science}}
  \bibinfo{volume}{1}(\bibinfo{number}{3}), pp. \bibinfo{pages}{237--267},
  \doi{10.1016/0304-3975(76)90059-1}.

\bibitemdeclare{article}{Halaby2016}
\bibitem{Halaby2016}
\bibinfo{author}{Mohamed~El \surnamestart Halaby\surnameend}
  (\bibinfo{year}{2016}): \emph{\bibinfo{title}{On the Computational Complexity
  of MaxSAT}}.
\newblock {\slshape \bibinfo{journal}{Electron. Colloquium Comput. Complex.}}
  \bibinfo{volume}{TR16}.
\newblock \urlprefix\url{https://api.semanticscholar.org/CorpusID:7872296}.

\bibitemdeclare{inproceedings}{IgnatievJM13}
\bibitem{IgnatievJM13}
\bibinfo{author}{Alexey \surnamestart Ignatiev\surnameend},
  \bibinfo{author}{Mikol{\'a}{\v{s}} \surnamestart Janota\surnameend} \&
  \bibinfo{author}{Joao \surnamestart Marques-Silva\surnameend}
  (\bibinfo{year}{2013}): \emph{\bibinfo{title}{Quantified Maximum
  Satisfiability}}.
\newblock In: {\slshape \bibinfo{booktitle}{Satisfiability Testing (SAT
  2013)}}, {\slshape \bibinfo{series}{LNCS}} \bibinfo{volume}{7962},
  \bibinfo{publisher}{Springer}, pp. \bibinfo{pages}{250--266},
  \doi{10.1007/978-3-642-39071-5\_19}.

\bibitemdeclare{article}{DBLP:journals/jcss/Immerman81}
\bibitem{DBLP:journals/jcss/Immerman81}
\bibinfo{author}{Neil \surnamestart Immerman\surnameend}
  (\bibinfo{year}{1981}): \emph{\bibinfo{title}{Number of Quantifiers is Better
  Than Number of Tape Cells}}.
\newblock {\slshape \bibinfo{journal}{J. Comput. Syst. Sci.}}
  \bibinfo{volume}{22}(\bibinfo{number}{3}), pp. \bibinfo{pages}{384--406},
  \doi{10.1016/0022-0000(81)90039-8}.

\bibitemdeclare{inproceedings}{Karp72}
\bibitem{Karp72}
\bibinfo{author}{Richard~M. \surnamestart Karp\surnameend}
  (\bibinfo{year}{1972}): \emph{\bibinfo{title}{Reducibility Among
  Combinatorial Problems}}.
\newblock In: {\slshape \bibinfo{booktitle}{Complexity of Computer
  Computations}}, \bibinfo{series}{The {IBM} Research Symposia Series},
  \bibinfo{publisher}{Plenum Press, New York}, pp. \bibinfo{pages}{85--103},
  \doi{10.1007/978-1-4684-2001-2\_9}.

\bibitemdeclare{article}{Kohli-et-al-94}
\bibitem{Kohli-et-al-94}
\bibinfo{author}{Rajeev \surnamestart Kohli\surnameend},
  \bibinfo{author}{Ramesh \surnamestart Krishnamurti\surnameend} \&
  \bibinfo{author}{Prakash \surnamestart Mirchandani\surnameend}
  (\bibinfo{year}{1994}): \emph{\bibinfo{title}{The Minimum Satisfiability
  Problem}}.
\newblock {\slshape \bibinfo{journal}{SIAM Journal on Discrete Mathematics}}
  \bibinfo{volume}{7}(\bibinfo{number}{2}), pp. \bibinfo{pages}{275--283},
  \doi{10.1137/S0895480191220836}.

\bibitemdeclare{inproceedings}{KupfermanS24}
\bibitem{KupfermanS24}
\bibinfo{author}{Orna \surnamestart Kupferman\surnameend} \&
  \bibinfo{author}{Noam \surnamestart Shenwald\surnameend}
  (\bibinfo{year}{2024}): \emph{\bibinfo{title}{Games with Weighted Multiple
  Objectives}}.
\newblock In: {\slshape \bibinfo{booktitle}{Automated Technology for
  Verification and Analysis ({ATVA} 2024)}}, {\slshape \bibinfo{series}{LNCS}}
  \bibinfo{volume}{15054}, \bibinfo{publisher}{Springer}, pp.
  \bibinfo{pages}{110--132}, \doi{10.1007/978-3-031-78709-6\_6}.

\bibitemdeclare{inproceedings}{PR89}
\bibitem{PR89}
\bibinfo{author}{A.~\surnamestart Pnueli\surnameend} \&
  \bibinfo{author}{R.~\surnamestart Rosner\surnameend} (\bibinfo{year}{1989}):
  \emph{\bibinfo{title}{On the synthesis of a reactive module}}.
\newblock In: {\slshape \bibinfo{booktitle}{Symposium on Principles of
  Programming Languages (POPL 1989)}}, \bibinfo{publisher}{Association for
  Computing Machinery}, p. \bibinfo{pages}{179–190},
  \doi{10.1145/75277.75293}.

\bibitemdeclare{inproceedings}{pnueli1977}
\bibitem{pnueli1977}
\bibinfo{author}{Amir \surnamestart Pnueli\surnameend} (\bibinfo{year}{1977}):
  \emph{\bibinfo{title}{The Temporal Logic of Programs}}.
\newblock In: {\slshape \bibinfo{booktitle}{Foundations of Computer Science
  (FOCS 1977)}}, \bibinfo{publisher}{{IEEE} Computer Society}, pp.
  \bibinfo{pages}{46--57}, \doi{10.1109/SFCS.1977.32}.

\bibitemdeclare{article}{SHARIR198167}
\bibitem{SHARIR198167}
\bibinfo{author}{M.~\surnamestart Sharir\surnameend} (\bibinfo{year}{1981}):
  \emph{\bibinfo{title}{A strong-connectivity algorithm and its applications in
  data flow analysis}}.
\newblock {\slshape \bibinfo{journal}{Computers {\&} Mathematics with
  Applications}} \bibinfo{volume}{7}(\bibinfo{number}{1}), pp.
  \bibinfo{pages}{67--72}, \doi{10.1016/0898-1221(81)90008-0}.

\bibitemdeclare{misc}{Zer13}
\bibitem{Zer13}
\bibinfo{author}{E.~\surnamestart Zermelo\surnameend} (\bibinfo{year}{1913}):
  \emph{\bibinfo{title}{Über eine Anwendung der Mengenlehre auf die Theorie
  des Schachspiels}}.

\bibitemdeclare{article}{ZIELONKA1998135}
\bibitem{ZIELONKA1998135}
\bibinfo{author}{Wieslaw \surnamestart Zielonka\surnameend}
  (\bibinfo{year}{1998}): \emph{\bibinfo{title}{Infinite games on finitely
  coloured graphs with applications to automata on infinite trees}}.
\newblock {\slshape \bibinfo{journal}{Theoretical Computer Science}}
  \bibinfo{volume}{200}(\bibinfo{number}{1}), pp. \bibinfo{pages}{135--183},
  \doi{10.1016/S0304-3975(98)00009-7}.

\end{thebibliography}
\end{document}